\documentclass[article,12pt]{article}

\usepackage{amsmath,amsthm,amssymb}
\usepackage{verbatim}
\usepackage{url}
\usepackage{fullpage}
\newtheorem{definition}{Definition}
\newtheorem{example}{Example}

\newtheorem{theorem}{Theorem}

\newtheorem{lemma}{Lemma}

\newcommand\blfootnote[1]{%
  \begingroup
  \renewcommand\thefootnote{}\footnote{#1}%
  \addtocounter{footnote}{-1}%
  \endgroup
  }

\begin{document}

\title{On the Linear Programming Bound for Lee-codes}

\author{\small Helena Astola, Ioan Tabus\\%\ead{ioan.tabus@tut.fi}
\small Department of Signal Processing, Tampere University of Technology\\
\small P.O.\ Box 553, FI-33101 Tampere, Finland \\
\small Email: helena.astola@tut.fi} \date{}
\maketitle

\begin{abstract}
Finding the largest code with a given minimum distance is one of the most basic problems in coding theory. In this paper, we study the linear programming bound for codes in the Lee metric. We introduce refinements on the linear programming bound for linear codes in the Lee metric, which give a tighter bound for linear codes. We also discuss the computational aspects of the problem and introduce a more compact program by the obtained refinements, and a recursive method for computing the so-called Lee-numbers, which are important in the computation of the linear programming bound.
\end{abstract}

\blfootnote{Date: June 13, 2014.

}

\section{Introduction}
\label{sec-intro}

Lee-codes were first introduced by C. Y. Lee in \cite{[Lee_someprop]}. The properties of Lee-codes and, in particular, the existence or  nonexistence of perfect codes in the Lee metric have been studied by numerous authors, for example, in \cite{[SolomonGolomb_lee]}  and \cite{[Astola-Leecodes]}, and more recently in \cite{[horak_lee]}, \cite{[horak_springer]}, \cite{[vardy_lee]} and \cite{[etzion_leecodes]}. However, the research and literature on Lee-codes is not extensive. In data transmission, the Lee metric can be used in phase modulation schemes, since corrupted digits of phase-modulated signals are more likely to have only slightly different phase than greatly different phase compared to the original signal \cite{[Berlekamp]}. There have been some more recent applications of Lee-codes to, for example, VLSI decoders and fault-tolerant logic, which are discussed in \cite{[leemetricBCH]}, \cite{[LeeBCH_decoding]} and \cite{[Ontheuse_Leecodes]}.

One of the most fundamental problems in coding theory is finding the largest code with a given minimum distance. The problem has been studied by several authors, in particular in the Hamming metric. The most well-known bounds are the Hamming bound, Plotkin bound, Singleton bound and Elias bound, and these bounds have been formulated for the Lee metric also. %, although the expressions are slightly more complicated.
The Hamming bound states that the total volume of the radius $t$ spheres around codewords of a code with minimum distance $2t+1$ is at most the volume of the entire space. Codes that meet the Hamming bound are called perfect codes. The Plotkin bound is based on the observation that the minimum distance between any pair of codewords cannot exceed the average distance between all pairs of codewords. The Singleton bound is for linear codes and is based on the observation that the minimum distance of a code cannot be greater than the minimum distance of any of its subcodes. The Elias bound basically combines the Hamming bound and the Plotkin bound to obtain a stronger bound for medium rates, since the Hamming bound is tight at high rates and the Plotkin bound is tight at low rates.

In \cite{[delsarte_thesis]}, Delsarte introduced association schemes to coding theory to deal with topics involving the distance distribution of a code. The theory and applications of association schemes into coding theory have been studied by numerous authors, and an extensive survey of these is given in \cite{[delsarte_levenshtein]}. An important approach to the problem of determining the upper bound for the size of a code is the linear programming approach, which follows from the association scheme structure in the Hamming metric \cite{[delsarte_thesis]}, \cite{[MacWilliams-Sloane]}. In fact, in the Hamming metric, the asymptotically best upper bound is the McEliece-Rodemich-Rumsey-Welch bound, see \cite{[mceliece-rodemich-]}, which is based on the linear programming approach. The Lee metric also forms an association scheme, although the structure is more complicated. In the Hamming metric, the distance relations between codewords directly define an association scheme, but not in the Lee-metric. Therefore, based on the linear programming bound, it is not possible to formulate simple expressions for bounds in the Lee metric. Generalizing to finite Frobenius rings, the linear programming bound has been studied recently in \cite{[byrne_frobenius]}.

For the Hamming metric, extensive tables of bounds on the size of codes for both binary and non-binary codes have been constructed, and such tables can be found, for example, in \cite{[Grassl:codetables]}. For the Lee metric, some values have been computed for small alphabets in \cite{[newupperbounds]}. In \cite{[astola_tabus_bounds]}, a recursive formula for computing the Lee-numbers needed in the computation of the linear programming bound was introduced and more extensive tables with presently known best upper bounds can be found in \cite{[astola:codetables]}.

In this paper, we discuss refinements on the linear programming bound for linear Lee-codes. In the Hamming metric, multiplying codewords by some constant does not change the weight of the codewords. However, in the Lee metric, multiplication typically changes the Lee-composition of the codeword and so also usually the Lee-weight. With linear codes, since they are linear subspaces of vector spaces, all the multiplied versions of any codeword also belong to the code. An important observation is that there must be as many codewords having the Lee-composition of a given codeword $\mathbf{x}$ as there are codewords having the Lee-composition of $r\mathbf{x}$ that is obtained by multiplication of the given codeword $\mathbf{x}$ by some constant $r$. Therefore, we get additional equality constraints between the coefficients of the weight distribution in the linear programming problem. Also, we can reduce the complexity of the problem since some coefficients can be assigned to zero following from the above observation. Furthermore, for a linear code, the linear transformation of the distribution vector by the second eigenmatrix of the scheme gives the distribution vector of the dual code. As it is also linear, the above observations also give constraints for the transformed vector, which turn out be equivalent to the constraints given by the linearity of the code. By introducing these equalities, we obtain a tighter bound for linear codes in the Lee metric.

The paper is organized as follows. Because the literature on Lee-codes is somewhat scattered and the notations vary, in Section \ref{sec-scheme} we review the concept of association schemes, in particular the Lee-scheme, and computing the linear programming bound for Lee-codes. Important concepts when determining the linear programming bound are Lee-compositions and Lee-numbers. In Section \ref{sec-refi}, we introduce the refinements on the linear programming bound for linear Lee-codes. In Section \ref{sec-compu}, we discuss the computational aspects in the linear programming problem. We discuss the computation of Lee-numbers, which can be done by recursion, and introduce an effective recursion based on the polynomial representation of the Lee-numbers. We also introduce a more compact linear program for solving the bounds for linear Lee-codes. We present the results obtained by applying the refinements in Section \ref{sec-result} together with some example linear codes that meet the obtained bounds.

\section{Lee Codes and Lee Schemes}
\label{sec-scheme}

In this section, we give a short but careful review of codes and schemes in the Lee metric. For basic properties of error-correcting codes we refer to \cite{[Berlekamp]}, \cite{[MacWilliams-Sloane]}.

Denote by $\mathbb{Z}_{q}^{n}$ the set of $n$-tuples with elements from the set $\{0,1,\ldots,q-1\}$, i.e.,
$$\mathbb{Z}_{q}^{n} = \{\mathbf{x}= [x_{1},\ldots,x_{n}] \mid x_{i}\in \{0,1,\ldots,q-1\}\}.$$

The elements of $\mathbb{Z}_{q}^{n}$ are $q$-ary vectors of block length $n$. A code $C$ is a subset of $\mathbb{Z}_{q}^{n}$. When $q=p^{k}$, where $p$ is prime, $\mathbb{Z}_{q}^{n}$ is a vector space over the field of $q$ elements, denoted by $\mathbb{F}_{q}^{n}$. In this paper, for simplicity, we discuss only fields $\mathbb{F}_{q}^{n}$, where $q$ is prime and the Lee-distance is naturally defined on the elements. $C$ is a linear code if it is a linear subspace of $\mathbb{F}_{q}^{n}$. The elements of $C$ are called codewords. A linear code $C$ of dimension $k\leq n$ is spanned by $k$ linearly independent vectors of $C$.

The Hamming distance $d_{H}(\mathbf{x},\mathbf{y})$ of vectors $\mathbf{x}$ and $\mathbf{y}$ of length $n$ is the number of coordinates where $\mathbf{x}$ and $\mathbf{y}$ differ, i.e. $d_{H}(\mathbf{x},\mathbf{y}) = |\{i \mid x_{i}\neq y_{i}\}|$.
The Hamming-weight $w_{H}$ of a vector $\mathbf{x}$ is $w_{H}(\mathbf{x})=d_{H}(\mathbf{0},\mathbf{x})$.

The Lee distance $d_{L}(\mathbf{x},\mathbf{y})$ of $q$-ary vectors $\mathbf{x}$ and $\mathbf{y}$ of length $n$ is defined as follows:
\begin{eqnarray}d_{L}(\mathbf{x},\mathbf{y})=\sum_{i = 1}^{n}\min(|x_{i}-y_{i}|,q-|x_{i}-y_{i}|).\end{eqnarray}

The Lee-weight $w_{L}$ of a vector $\mathbf{x}$ is $w_{L}(\mathbf{x})=d_{L}(\mathbf{0},\mathbf{x})$.

A code $C$ is $e$-error-correcting if the minimum distance between two codewords is $2e+1$. An $e$-Lee-error-correcting code will be able to detect and correct errors, which are at the Lee distance less than or equal to $e$ from the encoded sequence, i.e., have Lee-weight less than or equal to $e$.

An association scheme is a set together with relations defined on it that satisfy certain properties. The following definitions are given according to \cite{[Astola-Leecodes]}, \cite{[delsarte_thesis]},   \cite{[MacWilliams-Sloane]}:

\begin{definition}
An {\it association scheme with $n$ classes} consists of
a finite set $X$ together with $n+1$ relations
$R_{0},R_{1},\ldots,R_{n}$ defined on $X$ which satisfy
\item{(i)}
Each $R_{i}$ is symmetric: $(x, y) \in R_{i}\Rightarrow (y,x)\in
R_{i}$.
\item{(ii)} For every $x,y\in X, \ (x,y)\in R_{i}$ for exactly one $i$.
\item{(iii)} $R_{0}=\left\{(x,x)\mid x\in X\right\}$ is the identity relation.
\item{(iv)} If $(x,y)\in R_{k}$, the number of $z\in X$ such that $(x,z)\in R_{i}$ and $(y,z)\in R_{j}$ is a
constant $c_{ijk}$ depending on $i, j, k$ but not on the particular
choice of $x$ and $y$.
\end{definition}

For example, the Hamming scheme consists of the set of $q$-ary vectors of length $n$ and the vectors $\mathbf{x},\mathbf{y}\in R_{i}$ if their Hamming distance is $i$. It is easy to verify, that the above conditions hold for the Hamming scheme.

In order to define the Lee-scheme, first we need to define the Lee-composition of a vector. The Lee-composition
$l(\mathbf{x})$ of $\mathbf{x}\in \mathbb{Z}_{q}^{n}$ is the vector
\begin{eqnarray*}
l(\mathbf{x})=[l_{0}(\mathbf{x}),l_{1}(\mathbf{x}),\ldots,l_{s}(\mathbf{x})],
\end{eqnarray*}
where $s=\lfloor\frac{q}{2}\rfloor$ and $l_{i}(\mathbf{x})$ is the number of
the components of $\mathbf{x}$ having Lee-weight $i$.

Now, consider the $q$-ary vectors of length
$n = 1$. For this case, the distance relations define an association scheme in the Lee metric.
More formally, let $X=\mathbb{Z}_{q}$ and define the relations $R_{0},R_{1},\ldots,R_{s}, \
s=\lfloor\frac{q}{2}\rfloor$ with
\begin{eqnarray*}
(x,y) \in R_{i} \Leftrightarrow d_{L}(x,y)=i.
\end{eqnarray*}

The conditions (i)-(iv) of an association scheme can easily be shown to be
satisfied.

For $n>1$ the Lee-scheme is defined as the Delsarte extension of the one-dimensional Lee-scheme and, thus, forms an association scheme defined as follows. Take two elements $\mathbf{x}=[x_{1},\ldots,x_{n}]$,
$\mathbf{y}=[y_{1},\ldots,y_{n}]$ of $\mathbb{Z}_{q}^{n}$. Let $\rho_{t}(\mathbf{x},\mathbf{y})$ be
the number of integers $i$, $1\leq i\leq n$ such that
$(x_{i},y_{i})\in R_{t}$ % tämä tosiaan R_t eikö?
and define the following $(s + 1)$-tuple:
$$\rho(\mathbf{x},\mathbf{y})=[\rho_{0}(\mathbf{x},\mathbf{y}),\rho_{1}(\mathbf{x},\mathbf{y}),\ldots,\rho_{s}(\mathbf{x},\mathbf{y})].$$

Now $\rho(\mathbf{x},\mathbf{y})$ equals the Lee-composition of
the vector $\mathbf{x}-\mathbf{y}$. The number of distinct Lee-compositions is ${n+s\choose s}$.

Let $\rho^{(0)}=\rho(\mathbf{x},\mathbf{x})=[n,0,\ldots,0]$, i.e., the Lee-composition of the all zero vector, and denote by
$\rho^{(1)},\ldots,\rho^{(\alpha)}$, where $\alpha={n+s \choose
s}-1$, the other distinct Lee-compositions. Let us define the set
$K_{0},K_{1},\ldots,K_{\alpha}$ of relations on $\mathbb{Z}_{q}^{n}$ as follows
\begin{eqnarray*}
(\mathbf{x},\mathbf{y})\in K_{i} \Leftrightarrow
\rho(\mathbf{x},\mathbf{y})=\rho^{(i)}.
\end{eqnarray*}
In other words, $(\mathbf{x},\mathbf{y})\in K_{i}$ if the Lee-composition of the vector $\mathbf{x}-\mathbf{y}$ equals $\rho^{(i)}$.

Now $\mathbb{Z}_{q}^{n}$ together with the relations $K_{i}$ form an association
scheme of $\alpha$ classes. This will be called the
Lee-scheme.

The relations $K_{i}$ can be described by their adjacency matrices, i.e., matrices $D_{i}$ with rows and columns labeled by the points of $\mathbb{Z}_{q}^{n}$, where
$$(D_{i})_{\mathbf{x},\mathbf{y}} = \begin{cases} 1 & \hbox{ if } (\mathbf{x},\mathbf{y})\in K_{i},\\ 0 & \hbox{ otherwise.}\end{cases}$$

These adjacency matrices generate an associative and commutative algebra called the Bose-Mesner algebra of the association scheme.

Let $\xi = \exp(\frac{2\pi \sqrt{-1}}{q})$. When $\mathbf{t}=[t_{0},\ldots,t_{s}]$ and
$\mathbf{u}=[u_{0},\ldots,u_{s}]$ are Lee-compositions we define
the {\it Lee-numbers} $L_{\mathbf{t}}(\mathbf{u})$ \cite{[Astola-Leecodes]} from
\begin{eqnarray}
\label{eq-leenumb2s1}
&&\prod_{l=0}^{s} (z_{0}+(\xi^{l}+\xi^{-l})z_{1} +
(\xi^{2l}+\xi^{-2l})z_{2} +\cdots\\
&&+(\xi^{sl}+\xi^{-sl})z_{s})^{u_{l}}=
\sum_{\mathbf{t}}L_{\mathbf{t}} (\mathbf{u})z_{0}^{t_{0}}
\cdots z_{s}^{t_{s}},
\hbox{ for } q=2s+1\nonumber
\end{eqnarray}
and from
\begin{eqnarray}
\label{eq-leenumb2s}
&&\prod_{l=0}^{s} (z_{0}+(\xi^{l}+\xi^{-l})z_{1} +
(\xi^{2l}+\xi^{-2l})z_{2} +\cdots\nonumber\\
&&+(\xi^{(s-1)l}+\xi^{-(s-1)l})z_{s-1} +\xi^{sl}z_{s})^{u_{l}}\\ &&=\sum_{\mathbf{t}}L_{\mathbf{t}} (\mathbf{u})z_{0}^{t_{0}}
\cdots z_{s}^{t_{s}} \hspace{3.0cm}\hbox{ for } q=2s.\nonumber
\end{eqnarray}

Consider the Lee-scheme $\mathbb{Z}_{q}^{n}$ with the relations $K_{\mathbf{t}}$.
Let $C$ be a nonempty subset of $\mathbb{Z}_{q}^{n}$. The inner distribution of $C$
is the $(\alpha+1)$ -tuple of rational numbers
$B_{\mathbf{t}}$, where
\begin{eqnarray}
\label{eq-innerdist}
B_{\mathbf{t}} = \frac{1}{|C|} |K_{\mathbf{t}} \cap C^{2}|.
\end{eqnarray}

Now
\begin{eqnarray}
\label{eq-inner2}
B_{\mathbf{t}_{0}} =1, \ B_{\mathbf{t}}\geq 0 \hbox{ and }
\sum_{\mathbf{t}} B_{\mathbf{t}}=|C|.
\end{eqnarray}

For Lee-compositions $\mathbf{k}$
\begin{eqnarray}
B'_{\mathbf{k}} = \frac{1}{|C|} \sum_{\mathbf{t}}
L_{\mathbf{k}}(\mathbf{t}) B_{\mathbf{t}} \geq 0,
\end{eqnarray}
i.e., certain linear combinations
of the numbers $B_{\mathbf{t}}$ are nonnegative, which makes it possible
to apply the linear
programming bound to Lee-codes.
This can be proved using the properties of the association scheme, see \cite{[Astola-Leecodes]}.

Let us write for a composition $\mathbf{t}=[t_{0},t_{1},\ldots,t_{s}]$ \cite{[Astola-Leecodes]}:
$$\left[\begin{matrix}n\\ \mathbf{t}\end{matrix}\right]={n \choose \mathbf{t}}2^{n-t_{0}} \hspace{1.0cm} \hbox{ for } q=2s+1$$
and
$$\left[\begin{matrix}n\\ \mathbf{t}\end{matrix}\right]={n \choose \mathbf{t}}2^{n-t_{0}-t_{s}} \hspace{1.0cm} \hbox{ for } q=2s,$$
where ${n \choose \mathbf{t}}$ is the multinomial coefficient, which for nonnegative integers $t_{1},\ldots,t_{r}$ with $m=t_{1}+\cdots +t_{r}$ is defined as
$${m\choose t_{1},\ldots,t_{r}}=\frac{(t_{1}+\cdots+t_{r})!}{t_{1}!\cdots t_{r}!}=\frac{m!}{t_{1}!\cdots t_{r}!}.$$
The value of $\left[\begin{matrix}n\\ \mathbf{t}\end{matrix}\right]$ corresponds to the number of vectors in $\mathbb{Z}_{q}^{n}$ for which the Lee-composition is $\mathbf{t}$.
Let $C$ be a block code of length $n$ over $\mathbb{Z}_{q}$. The inner distribution
$[B_{\mathbf{t}_{0}},\ldots,B_{\mathbf{t}_{\alpha}}]$ of $C$ in the
corresponding Lee-scheme is
given by (\ref{eq-innerdist}). Due to (\ref{eq-inner2}) and since the eigenvalues of the Lee-scheme are given by the Lee-numbers we have
\begin{eqnarray}
|C|=\sum_{i=0}^{\alpha} B_{\mathbf{t}_{i}}, \ B_{\mathbf{t}}\geq0,
\ \mathbf{t}\in \left\{\mathbf{t}_{0},\ldots,\mathbf{t}_{\alpha}\right\}
\end{eqnarray}
and for any $\mathbf{k}$
\begin{eqnarray}
\sum_{i=1}^{\alpha}
L_{\mathbf{k}}(\mathbf{t}_{i})B_{\mathbf{t}_{i}} \geq -
\left[\begin{matrix} n \\ \mathbf{k}\end{matrix}\right].
\end{eqnarray}

Now, we may formulate the primal linear programming problem for the Lee-scheme \cite{[Astola-Leecodes]}:

\begin{theorem}\label{the-linprog}
Let $B_{\mathbf{t}}^{*}, \ \mathbf{t}\in \left\{\mathbf{t}_{1}, \ldots,\mathbf{t}_{\alpha}\right\}$ be an optimal
solution of the linear programming problem
\begin{eqnarray}
\label{eq-primal_lee}
&&\sum_{i=1}^{\alpha}B_{\mathbf{t}_{i}} =\max!\nonumber\\
&&B_{\mathbf{t}_{i}}\geq 0, \ i\in I \hbox{ and }B_{\mathbf{t}_{i}}= 0, \ i\in \{1,\ldots,\alpha\}\setminus I\nonumber\\
&&\sum_{i=1}^{\alpha} B_{\mathbf{t}_{i}}
L_{\mathbf{k}}(\mathbf{t}_{i})\geq - \left[\begin{matrix}n\\ \mathbf{k}\end{matrix}\right],
\ \\ &&\small \hbox{for all }\mathbf{k} \hbox{ running through Lee-compositions},\nonumber
\end{eqnarray}
where $I = \{i\mid l(\mathbf{x})=\mathbf{t}_{i} \hbox{ and } w_{L}(\mathbf{x})\geq d\}$. Then $1+\sum_{i=1}^{\alpha} B_{\mathbf{t}_{i}}^{*}$ is an upper bound to the size of the code $C$ with the minimum distance $d$.
\end{theorem}

\section{Refinements on the Linear Programming Bound for Linear Lee-Codes}
\label{sec-refi}

In the Hamming metric, multiplying codewords by some constant does not change the weight of the codewords. In the Lee metric, however, when a codeword is multiplied by some constant, the Lee-composition and the Lee-weight of the codeword are usually changed. Since linear codes are linear subspaces of vector spaces, all the vectors obtained by multiplying a codeword also belong to the code. In this section, we show that there must be as many codewords having the Lee-composition of a given codeword $\mathbf{x}$ as there are codewords having the Lee-composition of the codeword $r\mathbf{x}$ that is obtained by multiplying $\mathbf{x}$ by some constant $r$. This property can then be used to formulate constraining equalities into the linear programming problem, since the cardinality of a set of codewords having a given Lee-composition corresponds to a coefficient of the inner distribution of the code.

For simplicity, we let $q$ be prime, $\mathbb{F}_{q}^{n}=\{\mathbf{x}\mid x_{i}\in \mathbb{F}_{q},i=1,\ldots,n\}$. %, and do not discuss the vector spaces over fields with $q=p^{k}$ here. %Let $c(\mathbf{x})$ be the composition of $\mathbf{x}\in\mathbb{\mathbb{Z}}_{q}^{n}$, i.e., $c(\mathbf{x})=(c_{0}(\mathbf{x}),\ldots,c_{q-1}(\mathbf{x}))$, $c_{i}(\mathbf{x})=|\{j\mid x_{j}=i\}|$.
Let $C\subseteq \mathbb{F}_{q}^{n}$ be a linear code. For a linear code, if $\mathbf{x}\in C$, then $r\mathbf{x}\in C$ for all $r\in\{0,\ldots,q-1\}$. Denote by $C_{\mathbf{t}}$ the set of codewords having the Lee-composition $\mathbf{t}$.

Let us first examine how to obtain the Lee-composition of the vector $r\mathbf{x}$ from the Lee-composition of the vector $\mathbf{x}$. The vector $\mathbf{x}$ has the Lee-composition $l(\mathbf{x})=[l_{0}(\mathbf{x}),l_{1}(\mathbf{x}),\ldots,l_{s}(\mathbf{x})]$. The Lee-composition of the vector $r\mathbf{x}$ is clearly a permutation of the Lee-composition $l(\mathbf{x})$, since $r\mathbf{x}$ contains the same number of elements equal to $rx_{i}$ as the vector $\mathbf{x}$ contains elements equal to $x_{i}$. Thus, the Lee-composition of the vector $r\mathbf{x}$ is $l(r\mathbf{x})=[l_{\pi_{r}(0)}(\mathbf{x}),l_{\pi_{r}(1)}(\mathbf{x}),\ldots,l_{\pi_{r}(s)}(\mathbf{x})]$, where $\pi_{r}(i)=|k|$ such that $kr \equiv i \mod q$ and $-s\leq k \leq s$.

Now, we introduce the following Lemma:

\begin{lemma}
Given any Lee-composition $\mathbf{t}$ and any integer $r\in\{1,\ldots,q-1\}$, define the Lee-composition $\mathbf{u}=[t_{0},t_{\pi_{r}(1)},\ldots,t_{\pi_{r}(s)}]$. For any linear code the sets $C_{\mathbf{t}}$ and $C_{\mathbf{u}}$ have equal cardinalities.
\end{lemma}

\begin{proof}
First, any $\mathbf{x}_{1}\in C_{\mathbf{t}}$ has one corresponding element $r\mathbf{x}_{1}\in C_{\mathbf{u}}$. Also, every two distinct codewords in $C_{\mathbf{t}}$ correspond to two distinct elements in $C_{\mathbf{u}}$, since $r\mathbf{x}_{1}\neq r\mathbf{x}_{2}$ if $\mathbf{x}_{1}\neq \mathbf{x}_{2}$. Finally, every codeword in $C_{\mathbf{u}}$ corresponds to one codeword in $C_{\mathbf{t}}$, since for any $\mathbf{y}\in C_{\mathbf{u}}$, we may take $r^{-1}$ such that $r^{-1}r\equiv 1\mod q$ and $r^{-1}\mathbf{y}$ will belong to $C_{\mathbf{t}}$.
\end{proof}

The cardinalities of the sets $C_{\mathbf{t}}$ correspond to the coefficients of the inner distribution of the code $C$. Since there are equalities between these cardinalities, we get additional equality constraints in the linear programming problem for linear Lee-codes.

Let us denote by $\tau(\mathbf{t})$ the mapping that maps the Lee-composition $\mathbf{t}$ into the set of Lee-compositions, which are obtained from $\mathbf{t}$ by multiplication of vectors having the Lee-composition $\mathbf{t}$ by all $r\in\{1,\ldots,q-1\}$. Then,
\begin{eqnarray*}\tau(\mathbf{t}) = \{[t_{\pi_{r}(0)}(\mathbf{x}),t_{\pi_{r}(1)}(\mathbf{x}),\ldots,t_{\pi_{r}(s)}(\mathbf{x})]\mid \mathbf{t}=l(\mathbf{x}),\pi_{r}(i)=|k|,\\ kr \equiv i \mod q,-s\leq k \leq s, \ r,k\in\mathbb{F}_{q}\}.\end{eqnarray*}

We may now formulate the linear programming problem for linear Lee-codes:
\begin{theorem}\label{the-linprog-linref}
Let $B_{\mathbf{t}}^{*}, \ \mathbf{t}\in \left\{\mathbf{t}_{1}, \ldots,\mathbf{t}_{\alpha}\right\}$ be an optimal
solution of the linear programming problem
\begin{eqnarray}
\label{eq-primal_lee}
&&\sum_{i=1}^{\alpha}B_{\mathbf{t}_{i}} =\max!\nonumber\\
&&B_{\mathbf{t}_{i}}\geq 0, \ i\in I \hbox{ and }B_{\mathbf{t}_{i}}= 0, \ i\in \{1,\ldots,\alpha\}\setminus I\nonumber\\
&&B_{\mathbf{t}_{i}}=B_{\mathbf{t}_{j}} \hbox{ for all } \mathbf{t}_{j}\in \tau(\mathbf{t}_{i}) %\mathbf{t}_{i}=L(\mathbf{a}) \hbox{ and } \mathbf{t}_{j}=L(r\mathbf{a}) \hbox{ for some } r\in \mathbb{F}_{q}
\nonumber\\
&&\sum_{i=1}^{\alpha} B_{\mathbf{t}_{i}}
L_{\mathbf{k}}(\mathbf{t}_{i})\geq - \left[\begin{matrix}n\\ \mathbf{k}\end{matrix}\right],
\ \\ &&\small \hbox{for all }\mathbf{k} \hbox{ running through Lee-compositions},\nonumber
\end{eqnarray}
where $I = \{i\mid l(\mathbf{x})=\mathbf{t}_{i} \hbox{ and } w_{L}(\mathbf{x})\geq d\}$. Then $1+\sum_{i=1}^{\alpha} B_{\mathbf{t}_{i}}^{*}$ is an upper bound to the size of the code $C$ with the minimum distance $d$.
\end{theorem}

Let us look at an example on the equalities between different weight coefficients:

\begin{example}
Take the linear $[3,2]$-code over $\mathbb{F}_{7}$ with the generator matrix
\begin{eqnarray*}G = \left[\begin{matrix} 1 & 0& 2 \\ 0 & 1 & 4\end{matrix}\right].\end{eqnarray*}

The codewords of the code are
\begin{eqnarray*}[0,0,0],[0,1,4],[0,2,1],[0,3,5],[0,4,2],[0,5,6],[0,6,3],\\ \relax [1,0,2],[1,1,6],[1,2,3],[1,3,0],[1,4,4],[1,5,1],[1,6,5], \\ \relax [2,0,4],[2,1,1],[2,2,5],[2,3,2],[2,4,6],[2,5,3],[2,6,0], \\ \relax [3,0,6],[3,1,3],[3,2,0],[3,3,4],[3,4,1],[3,5,5],[3,6,2], \\ \relax [4,0,1],[4   ,  1  ,   5],[4   ,  2  ,   2],[4   ,  3  ,   6],[4   ,  4  ,   3],[4   ,  5  ,   0],[4   ,  6  ,   4], \\  \relax[5,  0  ,   3],[5   ,  1  ,   0],[5   ,  2  ,   4],[5   ,  3  ,   1],[5   ,  4  ,   5],[5   ,  5  ,   2],[5   ,  6  ,   6],\\  \relax[6   ,  0  ,   5],[6   ,  1  ,   2],[6   ,  2  ,   6],[6   ,  3  ,   3],[6   ,  4  ,   0],[6   ,  5  ,   4],[6   ,  6 ,    1].\end{eqnarray*}

The different Lee-compositions of the codewords are
\begin{eqnarray*}[3,0,0,0],[1,1,1,0],[1,1,0,1],[1,0,1,1],[0,3,0,0],[0,2,1,0],\\ \relax [0,1,1,1],[0,1,0,2],[0,0,3,0],[0,0,2,1],[0,0,0,3].\end{eqnarray*}

Let us take, for example, the second Lee-composition $[1,1,1,0]$. There are $6$ codewords having this Lee-composition, $[0,2,1]$, $[0,5,6]$, $[1,0,2]$, $[2,6,0]$, $[5,1,0]$ and $[6,0,5]$. Therefore, the coefficient of the inner distribution corresponding to the Lee-composition $[1,1,1,0]$ is $6$. If we multiply these vectors by $2$, we get the codewords $[0,4,2]$, $[0,3,5]$, $[2,0,4]$, $[4,5,0]$, $[3,2,0]$ and $[5,0,3]$, i.e., all such codewords, which have the Lee-composition $[1,0,1,1]$. Similarly, by multiplying with other possible values we obtain sets of codewords corresponding to certain Lee-compositions.

In Table \ref{tbl-compositions}, the Lee-compositions $\mathbf{t}_{i}$ of the code are listed together with the coefficients of the inner distribution $B_{\mathbf{t}_{i}}$. For each coefficient $B_{\mathbf{t}_{i}}$, the table shows also the set $\tau(\mathbf{t}_{i})$ of those coefficients of the inner distribution, which are equal to $B_{\mathbf{t}_{i}}$ under the transformations following from the linearity of the code, i.e., the coefficients of the inner distribution corresponding to such Lee-compositions, which are obtained from the Lee-compositions $\mathbf{t}_{i}$ by the mapping $\tau(\mathbf{t}_{i})$. The Lee-compositions are indexed based on the lexicographic order of all Lee-compositions for $\mathbb{F}_{7}^{3}$.

\begin{table*}[!h]\footnotesize
\caption{The Lee-compositions $\mathbf{t}_{i}$, the coefficients $B_{\mathbf{t}_{i}}$ and the set  $\tau(\mathbf{t}_{i})$ of the $[3,2]$-code over $\mathbb{F}_{7}$. The Lee-compositions are indexed based on the lexicographic order of all Lee-compositions for $\mathbb{F}_{7}^{3}$.}
\label{tbl-compositions}
 \begin{center}
\begin{tabular} {|c|c|c|c|}\hline
i & $\mathbf{t}_{i}$ & $B_{\mathbf{t}_{i}}$ & $\tau(\mathbf{t}_{i})$\\ \hline
0 & (3,0,0,0) & 1 & $\{B_{\mathbf{t}_{0}}\}$\\
5 & (1,1,1,0) & 6 & $\{B_{\mathbf{t}_{5}},B_{\mathbf{t}_{6}},B_{\mathbf{t}_{8}}\}$\\
6 & (1,1,0,1) & 6 & $\{B_{\mathbf{t}_{5}},B_{\mathbf{t}_{6}},B_{\mathbf{t}_{8}}\}$\\
8 & (1,0,1,1) & 6 & $\{B_{\mathbf{t}_{5}},B_{\mathbf{t}_{6}},B_{\mathbf{t}_{8}}\}$\\
10 & (0,3,0,0) & 2 & $\{B_{\mathbf{t}_{10}},B_{\mathbf{t}_{16}},B_{\mathbf{t}_{19}}\}$ \\
11 & (0,2,1,0) & 6 & $\{B_{\mathbf{t}_{11}},B_{\mathbf{t}_{15}},B_{\mathbf{t}_{17}}\}$\\
14 & (0,1,1,1) & 6 & $\{B_{\mathbf{t}_{14}}\}$\\
15 & (0,1,0,2) & 6 & $\{B_{\mathbf{t}_{11}},B_{\mathbf{t}_{15}},B_{\mathbf{t}_{17}}\}$\\
16 & (0,0,3,0) & 2 & $\{B_{\mathbf{t}_{10}},B_{\mathbf{t}_{16}},B_{\mathbf{t}_{19}}\}$\\
17 & (0,0,2,1) & 6 & $\{B_{\mathbf{t}_{11}},B_{\mathbf{t}_{15}},B_{\mathbf{t}_{17}}\}$\\
19 & (0,0,0,3) & 2 & $\{B_{\mathbf{t}_{10}},B_{\mathbf{t}_{16}},B_{\mathbf{t}_{19}}\}$\\ \hline
\end{tabular}
\end{center}
\end{table*}

\end{example}

\subsection{Dual Codes}
\label{ubsec-dual}

The MacWilliams identities state that the coefficients of the inner distribution of the dual of a code are given by a transformations of the coefficients of the inner distribution of the original code (for further reading, see, for instance \cite{[MacWilliams-Sloane]}). If we consider the above refinements with respect to the dual code, we may formulate more equality constraints, which follow from the connections between certain Lee-compositions in the dual code. Now \cite{[MacWilliams-Sloane]},
\begin{eqnarray}\label{eq-maw}\beta_{\mathbf{k}} = \frac{1}{|C|}\sum_{i=0}^{\alpha}L_{\mathbf{k}}(\mathbf{t}_{i})B_{\mathbf{t}_{i}},\end{eqnarray}
where $\beta_{\mathbf{k}}$ is a coefficient of the inner distribution of the dual code. Since the dual code is linear, some of these coefficients must be equal to each other, i.e., for some $u\neq v$, $\beta_{\mathbf{k}_{u}}=\beta_{\mathbf{k}_{v}}$. Hence, we may write
\begin{eqnarray*}\beta_{\mathbf{k}_{u}} = \frac{1}{|C|}\sum_{i=0}^{\alpha}L_{\mathbf{k}_{u}}(\mathbf{t}_{i})B_{\mathbf{t}_{i}} =  \frac{1}{|C|}\sum_{i=0}^{\alpha}L_{\mathbf{k}_{v}}(\mathbf{t}_{i})B_{\mathbf{t}_{i}} =\beta_{\mathbf{k}_{v}}.\end{eqnarray*}
We get the equality constraints
\begin{eqnarray*}\sum_{i=0}^{\alpha}L_{\mathbf{k}_{u}}(\mathbf{t}_{i})B_{\mathbf{t}_{i}} - \sum_{i=0}^{\alpha}L_{\mathbf{k}_{v}}(\mathbf{t}_{i})B_{\mathbf{t}_{i}} =
\sum_{i=0}^{\alpha}(L_{\mathbf{k}_{u}}(\mathbf{t}_{i})-L_{\mathbf{k}_{v}}(\mathbf{t}_{i}))B_{\mathbf{t}_{i}} =0.\end{eqnarray*}

Now, we may formulate the linear programming problem for linear Lee-codes with the equality constraints given by the dual code:
\begin{theorem}\label{the-linprog-dualref}
Let $B_{\mathbf{t}}^{*}, \ \mathbf{t}\in \left\{\mathbf{t}_{1}, \ldots,\mathbf{t}_{\alpha}\right\}$ be an optimal
solution of the linear programming problem
\begin{eqnarray}
\label{eq-primal_lee}
&&\sum_{i=1}^{\alpha}B_{\mathbf{t}_{i}} =\max!\nonumber\\
&&B_{\mathbf{t}_{i}}\geq 0, \ i\in I \hbox{ and }B_{\mathbf{t}_{i}}= 0, \ i\in \{1,\ldots,\alpha\}\setminus I\nonumber\\
%&&B_{\mathbf{t}_{i}}=B_{\mathbf{t}_{j}} \hbox{ iff } \mathbf{t}_{j}\in \tau(\mathbf{t}_{i}) \nonumber\\
&&\sum_{i=1}^{\alpha} B_{\mathbf{t}_{i}}L_{\mathbf{k}}(\mathbf{t}_{i})\geq - \left[\begin{matrix}n\\ \mathbf{k}\end{matrix}\right],\\
&&\small \hbox{for all }\mathbf{k} \hbox{ running through Lee-compositions},\nonumber \normalfont\\
&&\sum_{i=0}^{\alpha} B_{\mathbf{t}_{i}}(L_{\mathbf{k}_{u}}(\mathbf{t}_{i})-L_{\mathbf{k}_{v}}(\mathbf{t}_{i}))= 0 \hbox{ for all } \mathbf{k}_{v}\in \tau(\mathbf{k}_{u}),\nonumber
\end{eqnarray}
where $I = \{i\mid l(\mathbf{x})=\mathbf{t}_{i} \hbox{ and } w_{L}(\mathbf{x})\geq d\}$. Then $1+\sum_{i=1}^{\alpha} B_{\mathbf{t}_{i}}^{*}$ is an upper bound to the size of the code $C$ with the minimum distance $d$.
\end{theorem}

\begin{theorem} The problems in theorems \ref{the-linprog-linref} and \ref{the-linprog-dualref} are equivalent. \end{theorem}

\begin{proof} Rearrange now the indices of Lee-compositions so that in the sequence $\mathbf{t}_{0},\ldots,\mathbf{t}_{\alpha}$ the compositions $\tau(\mathbf{t})$ for which the coefficients $B_{\mathbf{t}}$ are constrained to have the same cardinalities will be in a consecutive order. Thus the valid solutions by the equality constraints given in Theorem \ref{the-linprog-linref} will be of the form
\begin{eqnarray*}
\left[\begin{matrix} B_{\mathbf{t}_{0}} \\ B_{\mathbf{t}_{1}} \\ \vdots \\B_{\mathbf{t}_{\alpha}}\end{matrix}\right] &=& \left[\begin{matrix} 1 & 0 & 0 & 0 &\cdots & 0 \\ 0& 1 & 0 & 0 &\cdots & 0  \\ \vdots & \vdots & \vdots & \vdots &\cdots & \vdots \\ 0 & 1 & 0 & 0 &\cdots & 0 \\ 0 & 0 & 1 & 0 &\cdots & 0  \\ \vdots & \vdots & \vdots & \vdots &\cdots & \vdots \\  &  &  &  & \ddots& \\ 0 & 0 & 0 & 0 &\cdots & 1 \end{matrix}\right]\left[\begin{matrix} \gamma_{0} \\ \gamma_{1} \\ \vdots \\ \gamma_{\kappa}\end{matrix}\right]\\ \mathbf{B}&=&A\mathbf{\gamma},\end{eqnarray*}
where the matrix $A$ has $\kappa+1$ blocks, each having as many rows as there are elements in the corresponding set $\tau(\mathbf{t}_{i})$.

By (\ref{eq-maw}), for any linear code
$$\boldsymbol{\beta}=\frac{1}{|C|}\Upsilon\mathbf{B},$$
where $\mathbf{B}$ and $\boldsymbol{\beta}$ are the inner distributions of the code and its dual code and $\Upsilon$ is the matrix containing the Lee-numbers with $\Upsilon(\mathbf{t},\mathbf{u})=L_{\mathbf{t}}(\mathbf{u})$.

Since $\mathbf{B}=A\boldsymbol{\gamma}$ and $\boldsymbol{\beta}=A\boldsymbol{\gamma}'$, we may write
\begin{eqnarray}\label{eq-upsilon}A\boldsymbol{\gamma}'=\frac{1}{|C|}\Upsilon A\boldsymbol{\gamma}.\end{eqnarray}

Now, we want to show that the above equation holds for any arbitrary $\boldsymbol{\gamma}$ in order to show the two problems equivalent. In other words, we want to show that the transformation of any vector of the form $A\boldsymbol{\gamma}$ by the Lee-numbers $\Upsilon$ is a vector of the form $A\boldsymbol{\gamma}'$.

Construct a linear code in $\mathbb{F}_{q}^{n}$ by taking a generator matrix having just one vector with a Lee-composition $\mathbf{t}_{i}$. The code has an inner distribution with nonzero values at position $B_{\mathbf{t}_{0}}$ and positions $B_{\mathbf{t}_{u}}$, $\mathbf{t}_{u}\in\tau(\mathbf{t}_{i})$. Hence, we obtain a vector $\boldsymbol{\gamma}_{i}$ with two nonzero values (with one nonzero value for the code having just the all-zero vector). Continue by taking another vector having a Lee-composition $\mathbf{t}_{j}\not\in\tau(\mathbf{t}_{i})$ as a generator matrix. Continue in such a way, for each Lee-composition not included in the previous sets $\tau$, so that we obtain $\kappa+1$ linearly independent vectors $\boldsymbol{\gamma}_{0},\ldots,\boldsymbol{\gamma}_{\kappa}$.

 Any arbitrary vector $\boldsymbol{\gamma}$ can now be obtained as a linear combination of the linearly independent vectors $\boldsymbol{\gamma}_{0},\ldots,\boldsymbol{\gamma}_{\kappa}$, i.e., we may take $\boldsymbol{\gamma}=a_{0}\boldsymbol{\gamma}_{0}+\cdots+ a_{\kappa}\boldsymbol{\gamma}_{\kappa}$. Since, the vectors $\boldsymbol{\gamma}_{0},\ldots,\boldsymbol{\gamma}_{\kappa}$ are those of linear codes, the equality in (\ref{eq-upsilon}) holds for them, i.e., we have for each $\boldsymbol{\gamma}_{i}$,
\begin{eqnarray}\label{eq-gamma}A\boldsymbol{\gamma}_{i}'=\frac{1}{|C|}\Upsilon A\boldsymbol{\gamma}_{i}.\end{eqnarray}

We may then write
\begin{eqnarray*}a_{0}\cdot A \boldsymbol{\gamma}_{0}'+\cdots+a_{\kappa}\cdot A \boldsymbol{\gamma}_{\kappa}'&=&a_{0}\cdot\frac{1}{|C|}\Upsilon  A \boldsymbol{\gamma}_{0}+\cdots+a_{\kappa}\cdot\frac{1}{|C|}\Upsilon  A\boldsymbol{\gamma}_{\kappa}\\
A\cdot( a_{0}\boldsymbol{\gamma}_{0}'+\cdots+ a_{\kappa}\boldsymbol{\gamma}_{\kappa}')&=&\frac{1}{|C|}\Upsilon A\cdot( a_{0}\boldsymbol{\gamma}_{0}+\cdots+ a_{\kappa}\boldsymbol{\gamma}_{\kappa})\\
A\boldsymbol{\gamma}'&=&\frac{1}{|C|}\Upsilon A\boldsymbol{\gamma}.\end{eqnarray*}

It remains to show that given any arbitrary $A\boldsymbol{\gamma}'$, the equation (\ref{eq-upsilon}) holds. We may construct $\boldsymbol{\gamma}'$ similarly as a linear combination of linearly independent vectors $\boldsymbol{\gamma}'_{0},\ldots,\boldsymbol{\gamma}'_{\kappa}$ constructed as the vectors $\boldsymbol{\gamma}_{0},\ldots,\boldsymbol{\gamma}_{\kappa}$ above. Hence, we may take $\boldsymbol{\gamma}'=b_{0}\boldsymbol{\gamma}'_{0}+\cdots+ b_{\kappa}\boldsymbol{\gamma}'_{\kappa}$. Again, for each $\boldsymbol{\gamma}_{i}'$ we have (\ref{eq-gamma}), and we may thus conclude using the above reasoning that the the equation (\ref{eq-upsilon}) holds.
\begin{comment}
\begin{eqnarray*}b_{0}\cdot A \boldsymbol{\gamma}_{0}'+\cdots+b_{\kappa}\cdot A \boldsymbol{\gamma}_{\kappa}'&=&b_{0}\cdot\frac{1}{|C|}\Upsilon  A \boldsymbol{\gamma}_{0}+\cdots+b_{\kappa}\cdot\frac{1}{|C|}\Upsilon  A\boldsymbol{\gamma}_{\kappa}\\
A\cdot( b_{0}\boldsymbol{\gamma}_{0}'+\cdots+ b_{\kappa}\boldsymbol{\gamma}_{\kappa}')&=&\frac{1}{|C|}\Upsilon A\cdot( b_{0}\boldsymbol{\gamma}_{0}+\cdots+ b_{\kappa}\boldsymbol{\gamma}_{\kappa})\\
A\boldsymbol{\gamma}'&=&\frac{1}{|C|}\Upsilon A\boldsymbol{\gamma}.\end{eqnarray*}
\end{comment}

\end{proof}

\section{Computational Aspects of the Linear Programming Problem for Lee-codes}
\label{sec-compu}

In this section, we discuss the computational aspects of the linear programming problem. Accurate results in the linear programming problem depend on efficient and accurate computation. We introduce a recursive way for computing the Lee-numbers, which play an important part in the computations. Based on the refinements on the problem for linear codes, we introduce a more compact linear program, where the set of linear constraints is reduced based on the theory. We can perform all computations with integers, resulting in very accurate results.

\subsection{Computing the Lee-numbers}

In \cite{[astola_tabus_bounds]}, a recursion for computing the Lee-numbers was introduced, providing the possibility of efficient computation of the bounds. We introduce here an alternative recursion following from the polynomial definition of the Lee-numbers. The recursion is based on the observation that as the length of the vector grows by one, it results in addition of $1$ in some component of the Lee-composition depending on the added component.

The equations (\ref{eq-leenumb2s1}) and (\ref{eq-leenumb2s}) give the Lee-numbers as coefficients of a generating polynomial. Let us now examine, how we can obtain them recursively using this generating polynomial by an example for $q=5$. Denote $\xi = \exp(\frac{2\pi \sqrt{-1}}{5})$. The Lee-numbers $L_{\mathbf{t}}(\mathbf{u})$ are now given by
\begin{eqnarray}
\label{eq-leenumbq5}
&&(z_{0}+2z_{1} +2z_{2})^{u_{0}}(z_{0}+(\xi+\xi^{-1})z_{1} +
(\xi^{2}+\xi^{-2})z_{2})^{u_{1}}\cdot\nonumber\\&&(z_{0}+(\xi^{2}+\xi^{-2})z_{1} +
(\xi+\xi^{-1})z_{2})^{u_{2}}=
\sum_{\mathbf{t}}L_{\mathbf{t}} (\mathbf{u})z_{0}^{t_{0}}z_{1}^{t_{1}}z_{2}^{t_{2}}.
\end{eqnarray}

Notice that if we denote $\zeta = \xi+\xi^{-1}=\xi+\xi^4$, then $\zeta^2=\xi^2+\xi^{-2}+2\xi\xi^{-1}=\xi^{2}+\xi^{3}+2$. Because $1+\xi+\xi^{2}+\xi^{3}+\xi^{4}=0$ we have $\zeta^{2}+\zeta-1=0$. Therefore, $\xi^2+\xi^{-2}=-1-\zeta$, %Further, as $\zeta^{2}=1-\zeta$, then $\zeta^{3}=\zeta-\zeta^{2}=\zeta-(1-\zeta)=2\zeta-1$, $\zeta^{4}=(1-\zeta)^{2}=1-2\zeta+(1-\zeta)=2-3\zeta$.
and we can write (\ref{eq-leenumbq5}) as
\begin{eqnarray}
\label{eq-leenumbq5_2}
&&(z_{0}+2z_{1} +2z_{2})^{u_{0}}(z_{0}+\zeta z_{1} +
(-1-\zeta)z_{2})^{u_{1}}\cdot\nonumber\\&&(z_{0}+(-1-\zeta)z_{1} +
\zeta z_{2})^{u_{2}}=
\sum_{\mathbf{t}}L_{\mathbf{t}} (\mathbf{u})z_{0}^{t_{0}}z_{1}^{t_{1}}z_{2}^{t_{2}}.
\end{eqnarray}

Assume that we have the Lee-numbers $L_{\mathbf{t}}(u_{0},u_{1},u_{2})$. Then for $L_{\mathbf{t}}(u_{0}+1,u_{1},u_{2})$ we have
\begin{eqnarray*}
&&\sum_{\mathbf{t}}L_{\mathbf{t}} (u_{0}+1,u_{1},u_{2})z_{0}^{t_{0}}z_{1}^{t_{1}}z_{2}^{t_{2}}
\\&&=(z_{0}+2z_{1}+2z_{2})\sum_{\mathbf{t}}L_{\mathbf{t}} (u_{0},u_{1},u_{2})z_{0}^{t_{0}}z_{1}^{t_{1}}z_{2}^{t_{2}}\\
&&=L_{(t_{0}-1,t_{1},t_{2})}(u_{0},u_{1},u_{2})z_{0}^{t_{0}}z_{1}^{t_{1}}z_{2}^{t_{2}}+ 2L_{(t_{0},t_{1}-1,t_{2})}(u_{0},u_{1},u_{2})z_{0}^{t_{0}}z_{1}^{t_{1}}z_{2}^{t_{2}}+\\
&&2L_{(t_{0},t_{1},t_{2}-1)}(u_{0},u_{1},u_{2})z_{0}^{t_{0}}z_{1}^{t_{1}}z_{2}^{t_{2}}.\end{eqnarray*}

Therefore,
\begin{eqnarray*}
&&L_{\mathbf{t}} (u_{0}+1,u_{1},u_{2}) = L_{(t_{0}-1,t_{1},t_{2})}(u_{0},u_{1},u_{2})+ \\ &&2L_{(t_{0},t_{1}-1,t_{2})}(u_{0},u_{1},u_{2})+
2L_{(t_{0},t_{1},t_{2}-1)}(u_{0},u_{1},u_{2}).\end{eqnarray*}

Similarly, for $L_{\mathbf{t}}(u_{0},u_{1}+1,u_{2})$ and $L_{\mathbf{t}}(u_{0},u_{1},u_{2}+1)$ we have
\begin{eqnarray*}
&&L_{\mathbf{t}} (u_{0},u_{1}+1,u_{2})=L_{(t_{0}-1,t_{1},t_{2})}(u_{0},u_{1},u_{2})+ \\ &&\zeta L_{(t_{0},t_{1}-1,t_{2})}(u_{0},u_{1},u_{2})+
(-1-\zeta)L_{(t_{0},t_{1},t_{2}-1)}(u_{0},u_{1},u_{2}),\end{eqnarray*}
and
\begin{eqnarray*}
&&L_{\mathbf{t}} (u_{0},u_{1},u_{2}+1)=L_{(t_{0}-1,t_{1},t_{2})}(u_{0},u_{1},u_{2})+ \\ &&(-1-\zeta)L_{(t_{0},t_{1}-1,t_{2})}(u_{0},u_{1},u_{2})+
\zeta L_{(t_{0},t_{1},t_{2}-1)}(u_{0},u_{1},u_{2}).\end{eqnarray*}

Notice, that the above recursions are of the form
$$L_{1}=L_{2}+(a+b\zeta)L_{3}+(c+d\zeta)L_{4},$$
and the possible initial values for $q=5$ with $n=1$ are exactly the coefficients appearing in these recursions, $\{1,2,\zeta,-1-\zeta\}$. Because
$$(a+b\zeta)(c+d\zeta)=ac+(ad+bc)\zeta+bd\zeta^{2}= ac+bd+(ad+bc-bd)\zeta,$$
we see that if we represent the Lee-numbers as vectors $[a,b]$ and define multiplication as $[a,b]\cdot[c,d]=[ac+bd,ad+bc+bd]$, we can perform all calculations with integers.

Consider now the case for $q=7$, where $\xi = \exp(\frac{2\pi \sqrt{-1}}{7})$. Denote again by $\zeta=\xi+\xi^{-1}$. Again, $\zeta^2=\xi^2+\xi^{-2}+2\xi\xi^{-1}=\xi^{2}+\xi^{3}+2$. For $q=7$, $1+\xi+\xi^{2}+\xi^{3}+\xi^{4}+\xi^{5}+\xi^{6}=0$, so we get for $\zeta^{3}=-\zeta^{2}+2\zeta+1$. The Lee-numbers $L_{\mathbf{t}}(\mathbf{u})$ are now given by
\begin{eqnarray*}
&&(z_{0}+2z_{1} +2z_{2}+2z_{3})^{u_{0}}(z_{0}+\zeta z_{1} +
(\zeta^{2}-2)z_{2}+(-\zeta^{2}-\zeta+1)z_{3})^{u_{1}}\cdot\nonumber\\&&(z_{0}+(\zeta^{2}-2)z_{1} +(-\zeta^{2}-\zeta+1)z_{2}+\zeta z_{3}
)^{u_{2}}\cdot\\ &&(z_{0}+(-\zeta^{2}-\zeta+1)z_{1}+\zeta z_{2} +(\zeta^{2}-2)z_{3})^{u_{3}}=
\sum_{\mathbf{t}}L_{\mathbf{t}} (\mathbf{u})z_{0}^{t_{0}}z_{1}^{t_{1}}z_{2}^{t_{2}}.\nonumber
\end{eqnarray*}

Hence, the recursions will be of the form
$$L_{1}=L_{2}+(a+b\zeta+c\zeta^{2})L_{3}+(d+e\zeta+f\zeta^{2})L_{4}+(g+h\zeta+i\zeta^{2})L_{5}.$$
The multiplication of two coefficients in the above equation is
\begin{eqnarray*}&&(a+b\zeta+c\zeta^{2})(d+e\zeta+f\zeta^{2})\\&&=ad+ae\zeta+af\zeta^{2}+bd\zeta+be\zeta^{2}+bf\zeta^{3}+cd\zeta^{2}+ce\zeta^{3}+cf\zeta^{4},\end{eqnarray*}
which, since $\zeta^{4}=3\zeta^{2}-\zeta-1$, results in the multiplication rule
$[a,b,c]\cdot[d,e,f]=[ad+bf+ce-cf,ae+bd+2ce+2bf-cf,af+be+cd-bf-ce+3cf]$ and we may again perform all computations with integers.

For the general $q$, the powers of $\zeta$ will be reduced according to the cyclotomic polynomial $\Phi_{q}(\xi)$, where $\xi$ are the roots of the cyclotomic polynomial, i.e., the primitive roots of unity $\xi = \exp(\frac{2\pi \sqrt{-1}}{q})$. This recursion provides very accurate values for the Lee-numbers, resulting in more accurate optimization in the linear programming problem.

\subsection{Compacting the Set of Linear Constraints}

Most linear programming solvers allow to express the constraints of the problems both in terms of inequality and equality constraints, thus the two formulations of the LP problem given in Theorems \ref{the-linprog-linref} and \ref{the-linprog-dualref} can easily be programmed and run. We examine the structure of the problem so that we can formulate it in a more compact form, leading to a faster execution.

We notice that by replacing the $\alpha$ variables $B_{\mathbf{t}_1},\ldots,B_{\mathbf{t}_\alpha}$ of the LP problem with the set of  variables $\gamma_1,\ldots,\gamma_{\kappa}$ we are eliminating the equality constraints from the LP problem. We introduce the vector $\boldsymbol{\gamma} =(\gamma_0,\gamma_1,\ldots,\gamma_\kappa)$ and formulate the equivalent  LP problem:
\begin{eqnarray}
&&\max_{\gamma_1,\ldots,\gamma_\kappa} \sum _{i=1}^\kappa |\tau_i |\gamma_i  \nonumber\\
 &&\ \ \ \mbox{subject to}\nonumber\\
  &&\ \ \  \gamma_0 =1  \mbox{ and }  \gamma_i \ge 0, i\in I \mbox{ and }  \gamma_i=0,i\in \{1,\ldots,\alpha\}\setminus I\nonumber\\
&&\ \ \ \Upsilon A \gamma \ge 0\nonumber
\end{eqnarray}
The cardinalities $|\tau_i |$ appear in the criterion of the problem since the initial criterion expressed in $\mathbf{B}$ is $\mathbf{1}^T \mathbf{B}$ (where $\mathbf{1}$ is the all one vector), and the criterion in the new variables is $\mathbf{1}^T A \boldsymbol{\gamma}$, where the new vector of coefficients, $\mathbf{1}^T A$, will have as elements the size of the partitions of $A$, which are equal to the cardinalities of the sets $\tau_i$.

  Additionally we notice that the matrix ${\cal U} = \Upsilon A$ can be seen to have the partition structure similar to that of $A$,
 \begin{eqnarray*}
  {\cal U}&=& \Upsilon A = \Upsilon \left[\begin{matrix}1 & 0 & 0 & 0 & \ldots & 0\\\hline
  0 & 1 & 0 & 0 & \ldots & 0\\ 0 & 1 & 0 & 0 & \ldots & 0\\ 0 & 1 & 0 & 0 & \ldots & 0\\ 0 & 1 & 0 & 0 & \ldots & 0 \\\hline
   0 & 0 & 1 & 0 & \ldots & 0  \\ 0 & 0 & 1 & 0 & \ldots & 0  \\\hline\ldots\\\hline  0 & 0 & 0 & 0 & \ldots & 1 \end{matrix}\right]    =
   \left[\begin{matrix}\Phi_{1,1} &\Phi_{1,2} & \Phi_{1,3} & \Phi_{1,4} & \ldots & \Phi_{1,\kappa}\\\hline
 \Phi_{2,1} &\Phi_{2,2} & \Phi_{2,3} & \Phi_{2,4} & \ldots & \Phi_{2,\kappa}\\ \Phi_{2,1} &\Phi_{2,2} & \Phi_{2,3} & \Phi_{2,4} & \ldots & \Phi_{2,\kappa}\\ \Phi_{2,1} &\Phi_{2,2} & \Phi_{2,3} & \Phi_{2,4} & \ldots & \Phi_{2,\kappa}\\ \Phi_{2,1} &\Phi_{2,2} & \Phi_{2,3} & \Phi_{2,4} & \ldots & \Phi_{2,\kappa} \\\hline
   \Phi_{3,1} &\Phi_{3,2} & \Phi_{3,3} & \Phi_{3,4} & \ldots & \Phi_{3,\kappa}  \\   \Phi_{3,1} &\Phi_{3,2} & \Phi_{3,3} & \Phi_{3,4} & \ldots & \Phi_{3,\kappa}    \\\hline\ldots\\\hline  \Phi_{\kappa,1} &\Phi_{\kappa,2} & \Phi_{\kappa,3} & \Phi_{\kappa,4} & \ldots & \Phi_{\kappa,\kappa} \end{matrix}\right]\\
 &=& A \Phi\end{eqnarray*}
   where the matrix $\Phi=[\Phi_{i,j}]$ is $\kappa\times \kappa$ and the rows in a partition corresponds to the Lee compositions belonging to the same set $\tau_i$. Inside a partition the rows of the matrix  ${\cal U}$ are identical, leading to a repeated inequality constraint. In order to remove this redundancy, we are selecting from the matrix ${\cal U}$ only one row per partition block, keeping thus only the non-redundant inequalities, resulting in a matrix ${\cal U}^\prime$, % having about $s$ times less rows (most of the blocks have $s$ rows)
   and we replace the inequality constraints $\Upsilon A \boldsymbol{\gamma} \ge 0$ with ${\cal U}^\prime\boldsymbol{\gamma} \ge 0$.

We also observe that the matrix ${\cal U}$ is formed of integer numbers, as opposed to the matrix $\Upsilon$ where the elements are various combinations of the powers of $\xi$, and are in most cases irrational numbers. This can be seen as follows.

The elements of ${\cal U}$ are sums of $k$ Lee-numbers:
 \begin{eqnarray}\label{eq-nums}L_{\mathbf{t_{1}}}(\mathbf{u})+L_{\mathbf{t_{2}}}(\mathbf{u})+ \cdots +L_{\mathbf{t_{k}}}(\mathbf{u})\hspace{4.0cm}\nonumber\\
 = \sum_{\mathbf{x}\mid L(\mathbf{x})=\mathbf{t}_{1}}\left(\prod_{i=1}^{n}\xi^{v_{i}x_{i}}\right) + \sum_{\mathbf{x}\mid L(\mathbf{x})=\mathbf{t}_{2}} \left(\prod_{i=1}^{n}\xi^{v_{i}x_{i}}\right) +\cdots + \sum_{\mathbf{x}\mid L(\mathbf{x})=\mathbf{t}_{k}} \left(\prod_{i=1}^{n}\xi^{v_{i}x_{i}}\right) \end{eqnarray}
 where $k$ is the cardinality of $\tau(\mathbf{t}_{i})$.

Since the Lee-numbers $L_{\mathbf{t}_{i}}(\mathbf{u})$ correspond to compositions according to $\tau$, then for each vector $\mathbf{x}$, in (\ref{eq-nums}) are also included all the vectors $r\mathbf{x}$, where $r\in\{1,\ldots,q-1\}$. Also, each vector can only have one Lee-composition and thus appear in only one of the sums in (\ref{eq-nums}).

 We may now rearrange and group the sum in (\ref{eq-nums}) according to these multiplications:
$$\left(\xi^{\mathbf{v}\cdot\mathbf{x}_{1}}+\xi^{\mathbf{v}\cdot 2\mathbf{x}_{1}}+ \cdots + \xi^{\mathbf{v}\cdot (q-1)\mathbf{x}_{1}}\right) + \left(\xi^{\mathbf{v}\cdot\mathbf{x}_{2}}+\xi^{\mathbf{v}\cdot 2\mathbf{x}_{2}}+ \cdots + \xi^{\mathbf{v}\cdot (q-1)\mathbf{x}_{2}}\right)+\cdots$$
This forms a partition of the set of vectors having a Lee-composition in $\tau(\mathbf{t}_{i})$, since the relation $\mathcal{R}$ defined as $(\mathbf{x},\mathbf{y})\in\mathcal{R}_{x}$ iff $\mathbf{x}=r\mathbf{y}$, $r\in\{1,\ldots,q-1\}$ is clearly an equivalence relation.

Therefore, we can group the sum into $m$ parts, each having $q-1$ terms. If we now look at one such part:
 $$\xi^{\mathbf{v}\cdot\mathbf{x}_{i}}+\xi^{\mathbf{v}\cdot 2\mathbf{x}_{i}}+ \cdots + \xi^{\mathbf{v}\cdot (q-1)\mathbf{x}_{i}},$$
 we see that, since $\mathbf{x}\cdot r\mathbf{y}=r(\mathbf{x}\cdot\mathbf{y})$, we have
 $$\xi^{\mathbf{v}\cdot\mathbf{x}_{i}}+\xi^{2(\mathbf{v}\cdot \mathbf{x}_{i})}+ \cdots + \xi^{(q-1)\mathbf{v}\cdot \mathbf{x}_{i}},$$
 where each exponent in the above sum is distinct. If the dot product $\mathbf{v}\cdot\mathbf{x}_{i}$ is $0$, the above sum is $q-1$. Otherwise, we have a sum of the form
 $$\xi+\xi^{2}+ \cdots + \xi^{q-1}=-1.$$

 So (\ref{eq-nums}) is a sum of the form $m_{1}(q-1)+m_{2}(-1)$, where $m_{1}$ and $m_{2}$ are integers $\geq 0$ such that $m=m_{1}+m_{2}$.

We may now use the more compact version of the linear programming problem performing computations only on integers, obtaining faster execution and very accurate results.

\section{Results}
\label{sec-result}

In Tables \ref{tbl-bounds_lin}-\ref{tbl-bounds-lin7} are the results for the upper bound of the parameter $k$ for linear Lee-codes with $q=5$ and $q=7$, obtained using the linear programming bound for linear codes. The most interesting cases are the situations where the general linear programming bound would allow for a linear code to exist with some parameter $k$ but the refinement shows that such a code cannot exist. For example, with $q=5$, $n=8$ and $d=8$ the linear programming bound is $134$ \cite{[astola_tabus_bounds]}, which does not deny the existence of a linear code with $k=3$. However, the refinement gives a bound of $75$, which shows that there cannot be a code with the parameter $k=3$. Another example would be for $q=7$ when $n=7$ and $d=11$. The linear programming bound gives a bound of 55, but with the refinement for linear codes the value 40 is obtained, implying that a linear code with $k=2$ cannot exist with these parameters.

In Tables \ref{tbl-bounds_lin}-\ref{tbl-bounds-lin7} bounds that were found to be tight are also shown. This was concluded by checking the minimum distances of linear codes generated randomly with given parameters $q$, $n$ and $k$.

\begin{table*}[!h]\footnotesize
\caption{Upper bounds for the parameter $k$ of linear Lee-codes when $q=5$. The $*$ indicates a tight bound and bold an improvement comparing to linear programming bounds in \cite{[astola_tabus_bounds]}.}
\label{tbl-bounds_lin}
 \begin{center}
\begin{tabular} {|c|c|c|c|c|c|c|c|c|c|c|c|c|c|c|c|c|c}\hline
$n\backslash d$ & 3 & 4 & 5& 6& 7 & 8& 9 & 10 & 11 & 12& 13 & 14-15 \\ \hline
2 & $1*$ &  &  & & & & & & & &  &    \\ \hline
3 & $1*$ & $1*$ & &  &  & & & & & &  &    \\ \hline
4 & $2*$ & $2*$ & $1*$ & $1*$ &  &  &  & & &&  &   \\ \hline
5 & $3*$ & $3*$ & $2*$ & $1*$ & $1*$  &  &  &  & & &  &    \\ \hline
6 & $4*$& $3*$ & $3*$ & $2*$ & $\mathbf{1*}$ & $1*$ & $1*$ & &  &  &  &  \\ \hline
7 & $5*$ & $4*$ & $3*$ & $3*$ & $2*$ & $\mathbf{1*}$ & $1*$ & $1*$& & &  &   \\ \hline
8 &  $6*$ & $5*$ & $4*$ & $4*$ & $3*$ & $\mathbf{2*}$ & $2*$ & $1*$ & $1*$&$1*$ &  &  \\ \hline
9 & $7*$ & $6*$ & $5*$ & $5$ & $4$ & $3*$ & $3$ & $2*$ & $\mathbf{1*}$ & $1*$ & $1*$&    \\ \hline
10 & $8*$ &$7*$& $6*$ & $6$ & $5$ & $4$ & $\mathbf{3*}$ & $3$ & $2*$& $2*$ & $1*$ & $1*$ \\ \hline
\end{tabular}
\end{center}
\end{table*}

\begin{table*}[!h]\footnotesize
\caption{Upper bounds for the parameter $k$ of linear Lee-codes when $q=7$. The $*$ indicates a tight bound and bold an improvement comparing to linear programming bounds in \cite{[astola_tabus_bounds]}.}
\label{tbl-bounds-lin7}
 \begin{center}
\begin{tabular} {|c|c|c|c|c|c|c|c|c|c|c|c|c|c|c|c|c|c|c|c|c|c|c|c|c|c|c|c|c|}\hline
$n\backslash d$&  3 & 4 & 5& 6& 7 & 8& 9 & 10 & 11 & 12& 13 & 14& 15 & 16 & 17 & 18 \\ \hline
2 &$1*$ &	&	 &	  & & & & & & & & & & & &\\ \hline
3 &$2*$ & 	$1*$ &	$1*$& 	$1*$& 	&	& & 	& & & && & & & \\ \hline
4 & $2*$ &	$2*$ &	$2*$& 	$1*$& 	$1*$& &	& 	 &	& 	& 	 & & 	 & & &    \\ \hline
5 & $3*$ 	&$3*$ &	$2*$ &	$2*$& 	$1*$& 	$1*$& 	$1*$ &	& &	& &	 &	 & & &   \\ \hline
6 & $4*$ 	&$4*$ &	$3*$ &	$3*$ 	&$2*$ &	$2*$ &	$\mathbf{1*}$& 	$1*$& 	 $1*$& 	 $1*$  & & & & & &\\ \hline
7 & $5*$ 	&$5*$ 	&$4*$ 	&$4$ 	&$3*$ &	$3$ 	&$2*$ 	&$2*$ 	 &$\mathbf{1*}$ 	 &$1*$ 	 &$1*$ & & & & &\\  \hline
8 & $6*$ & $\mathbf{6*}$ & $5*$ & $\mathbf{4*}$ & $4$ & $4$ & $3*$ & $3$ & $2*$ & $2*$ & $\mathbf{1*}$ & $1*$  & $1*$ & & & \\ \hline
9 & $7*$ & $\mathbf{6*}$ & $6$ & $5*$ & $5$ & $4*$ & $4$ & $3*$ & $3$ & $3$ & $2*$ & $\mathbf{1*}$ & $\mathbf{1*}$ & $1*$ & $1*$ & $1*$ \\ \hline
\end{tabular}
\end{center}
\end{table*}

\subsection{Codes meeting the Bounds}

In the following, we give some examples of codes, that meet the bounds obtained for linear codes.

Consider the bound for $k$ given in Table \ref{tbl-bounds_lin} with $q=5$, $n=8$ and $d=8$, which is $2$. This means that the maximum number of codewords in a linear code with these parameters is at most $25$. The code having the generator matrix
\begin{eqnarray*}G_{1}=\left[\begin{matrix} 1&0&0&2&2&3&3&1\\ 0&1&2&3&0&3&4&3\end{matrix}\right]\end{eqnarray*}
is a $[8,2]$-code with the minimum distance $8$, therefore, it is an optimal linear code for the above parameters.

Consider the bound for $k$ given in Table \ref{tbl-bounds_lin} with $q=5$, $n=9$ and $d=5$, which is $5$. This means that the maximum number of codewords in a linear code with these parameters is at most $3125$. The code having the generator matrix
\begin{eqnarray*}G_{2}=\left[\begin{matrix} 1&0&0&0&0&4&2&0&1\\ 0&1&0&0&0&2&4&1&0\\ 0&0&1&0&0&2&2&1&1\\0&0&0&1&0&3&3&3&1\\0&0&0&0&1&1&2&2&2\end{matrix}\right]\end{eqnarray*}
is a $[9,5]$-code with the minimum distance $5$, therefore, it is an optimal linear code for the above parameters.

Consider the bound for $k$ given in Table \ref{tbl-bounds-lin7} with $q=7$, $n=7$ and $d=5$, which is $4$. This means that the maximum number of codewords in a linear code with these parameters is at most $2401$. The code having the generator matrix
\begin{eqnarray*}G_{2}=\left[\begin{matrix} 1&0&0&0&5&4&4\\ 0&1&0&0&3&6&6\\ 0&0&1&0&1&4&6\\0&0&0&1&6&5&3\end{matrix}\right]\end{eqnarray*}
is a $[7,4]$-code with the minimum distance $5$, therefore, it is an optimal linear code for the above parameters.

\section{Conclusions}

In this paper, we introduced refinements on the linear programming bound for linear Lee-codes. These refinements are based on the observation that in the Lee metric, the multiplication of codewords typically changes the Lee-composition of the codeword and so also usually the Lee-weight. Therefore, and since the codes are linear, we are able to generate a mapping between the Lee-compositions, which follows in equalities between the coefficients of the inner distribution of the code, and, thus, in additional equality constraints in the linear programming problem. This refinement results in tighter bounds for linear Lee-codes.

We also discussed the computational aspects of the linear programming problem, including the computation of Lee-numbers, which can be done by recursion, and introduced an effective recursion based on the polynomial representation of the Lee-numbers. We introduced also a more compact form of the linear programming problem based on the refinements on linear Lee-codes. Our method is very accurate, since it allows all computations to be performed with integers.

\bibliographystyle{IEEEtran}
\bibliography{biblio}

\begin{thebibliography}{10}
\providecommand{\url}[1]{#1}
\csname url@samestyle\endcsname
\providecommand{\newblock}{\relax}
\providecommand{\bibinfo}[2]{#2}
\providecommand{\BIBentrySTDinterwordspacing}{\spaceskip=0pt\relax}
\providecommand{\BIBentryALTinterwordstretchfactor}{4}
\providecommand{\BIBentryALTinterwordspacing}{\spaceskip=\fontdimen2\font plus
\BIBentryALTinterwordstretchfactor\fontdimen3\font minus
  \fontdimen4\font\relax}
\providecommand{\BIBforeignlanguage}[2]{{%
\expandafter\ifx\csname l@#1\endcsname\relax
\typeout{** WARNING: IEEEtran.bst: No hyphenation pattern has been}%
\typeout{** loaded for the language `#1'. Using the pattern for}%
\typeout{** the default language instead.}%
\else
\language=\csname l@#1\endcsname
\fi
#2}}
\providecommand{\BIBdecl}{\relax}
\BIBdecl

\bibitem{[Lee_someprop]}
C.~Lee, ``Some properties of nonbinary error-correcting codes,'' \emph{IRE
  Transactions on Information Theory}, vol.~4, no.~2, pp. 77 --82, June 1958.

\bibitem{[SolomonGolomb_lee]}
S.~W. Golomb and L.~R. Welch, ``Perfect codes in the {L}ee metric and the
  packing of polyominoes,'' \emph{SIAM Journal on Applied Mathematics},
  vol.~18, no.~2, pp. 302--317, 1970.

\bibitem{[Astola-Leecodes]}
J.~Astola, ``The theory of {L}ee-codes,'' Lappeenranta University of
  Technology, Department of Physics and Mathematics, Research Report 1/1982.

\bibitem{[horak_lee]}
P.~Horak, ``On perfect {L}ee codes,'' \emph{Discrete Mathematics}, vol. 309,
  no.~18, pp. 5551--5561, September 2009.

\bibitem{[horak_springer]}
C.~Araujo, I.~Dejter, and P.~Horak, ``\BIBforeignlanguage{English}{A
  generalization of {L}ee codes},'' \emph{\BIBforeignlanguage{English}{Designs,
  Codes and Cryptography}}, vol.~70, no. 1-2, pp. 77--90, 2014.

\bibitem{[vardy_lee]}
T.~Etzion, A.~Vardy, and E.~Yaakobi, ``Dense error-correcting codes in the
  {L}ee metric,'' in \emph{Information Theory Workshop (ITW), 2010 IEEE},
  September 2010, pp. 1--5.

\bibitem{[etzion_leecodes]}
T.~Etzion, ``Product constructions for perfect {L}ee codes,'' \emph{IEEE
  Transactions on Information Theory}, vol.~57, no.~11, pp. 7473--7481,
  November 2011.

\bibitem{[Berlekamp]}
E.~R. Berlekamp, \emph{{Algebraic Coding Theory}}.\hskip 1em plus 0.5em minus
  0.4em\relax New York: McGraw-Hill, 1968.

\bibitem{[leemetricBCH]}
R.~Roth and P.~Siegel, ``Lee-metric {B}{C}{H} codes and their application to
  constrained and partial-response channels,'' \emph{Information Theory, IEEE
  Transactions on}, vol.~40, no.~4, pp. 1083 --1096, July 1994.

\bibitem{[LeeBCH_decoding]}
Y.~Wu and C.~Hadjicostis, ``Decoding algorithm and architecture for {B}{C}{H}
  codes under the {L}ee metric,'' \emph{IEEE Transactions on Communications},
  vol.~56, no.~12, pp. 2050--2059, December 2008.

\bibitem{[Ontheuse_Leecodes]}
H.~Astola and S.~Stankovi\'{c}, ``On the use of {L}ee-codes for constructing
  multiple-valued error-correcting decision diagrams,'' in \emph{5th
  International Symposium on Communications, Control, and Signal Processing},
  Rome, Italy, May 2-4, 2012.

\bibitem{[delsarte_thesis]}
P.~Delsarte, ``An algebraic approach to the association schemes of coding
  theory,'' Philips Res. Repts. Suppl., 1973.

\bibitem{[delsarte_levenshtein]}
P.~Delsarte and V.~Levenshtein, ``Association schemes and coding theory,''
  \emph{IEEE Transactions on Information Theory}, vol.~44, no.~6, pp. 2477
  --2504, October 1998.

\bibitem{[MacWilliams-Sloane]}
F.~J. MacWilliams and N.~J.~A. Sloane, \emph{{The Theory of Error-Correcting
  Codes}}.\hskip 1em plus 0.5em minus 0.4em\relax Amsterdam: North-Holland,
  1997.

\bibitem{[mceliece-rodemich-]}
R.~McEliece, E.~Rodemich, H.~Rumsey, and L.~Welch, ``New upper bounds on the
  rate of a code via the {D}elsarte-{M}ac{W}illiams inequalities,'' \emph{IEEE
  Transactions on Information Theory}, vol.~23, no.~2, pp. 157--166, March
  1977.

\bibitem{[byrne_frobenius]}
E.~Byrne, M.~Greferath, and M.~O’Sullivan, ``\BIBforeignlanguage{English}{The
  linear programming bound for codes over finite {F}robenius rings},''
  \emph{\BIBforeignlanguage{English}{Designs, Codes and Cryptography}},
  vol.~42, no.~3, pp. 289--301, 2007.

\bibitem{[Grassl:codetables]}
M.~Grassl, ``{Bounds on the minimum distance of linear codes and quantum
  codes},'' Online available at http://www.codetables.de, 2007.

\bibitem{[newupperbounds]}
J.~Quistorff, ``New upper bounds on {L}ee codes,'' \emph{Discrete Applied
  Mathematics}, vol. 154, no.~10, pp. 1510 -- 1521, 2006.

\bibitem{[astola_tabus_bounds]}
H.~Astola and I.~Tabus, ``Bounds on the size of {L}ee-codes,'' in \emph{8th
  International Symposium on Image and Signal Processing and Analysis},
  Trieste, Italy, September 2013, pp. 464--469,
  \url{http://www.cs.tut.fi/%7Eastola/0091-0144.pdf}.

\bibitem{[astola:codetables]}
H.~Astola, ``{Bounds on the size of Lee-codes},'' Online available at
  {\url{http://www.cs.tut.fi/%7Eastola/leecodetables.html}}, 2013.

\end{thebibliography}

\end{document}